\documentclass
[aps,preprint,preprintnumbers,amsmath,amssymb,nofootinbib]{revtex4}%
\usepackage{amssymb}
\usepackage{amsmath}
\usepackage{bm}
\usepackage{amsthm}
\usepackage{amsfonts}
\usepackage{graphicx}%
\providecommand{\U}[1]{\protect\rule{.1in}{.1in}}
\theoremstyle{plain}
\newtheorem{thm}{Theorem}
\theoremstyle{definition}
\newtheorem{defn}{Definition}

\theoremstyle{proposition}
\newtheorem{prop}{Proposition}
\theoremstyle{lemma}
\newtheorem{lemma}{Lemma}
\theoremstyle{corollary}
\newtheorem{coro}{Corollary}
\begin{document}
\title{On the embedding of Weyl manifolds }

\begin{abstract}
We discuss the possibility of extending different versions of the
Campbell-Magaard theorem, which have already been established in the context
of semi-Riemannian geometry, to the context of Weyl's geometry. We show that
some of the known results can be naturally extended to the new geometric
scenario, although new difficulties arise. In pursuit of solving the embedding
problem we have obtained some no-go theorems. We also highlight some of the
difficulties that appear in the embedding problem, which are typical of the
Weylian character of the geometry. The establishing of these new results may
be viewed as part of a program that highlights the possible significance of
embedding theorems of increasing degrees of generality in the context of
modern higher-dimensional space-time theories.

\end{abstract}
\author{R. Avalos, F. Dahia, C. Romero}
\affiliation{Departamento de F\'{\i}sica, Universidade Federal da Para\'{\i}ba, Caixa
Postal 5008, 58059-970 Jo\~{a}o Pessoa, PB, Brazil }
\affiliation{E-mail: ravalos@gmail.com; fdahia@fisica.ufpb.br; cromero@fisica.ufpb.br }

\maketitle

\section{Introduction}

The unification of the fundamental forces of nature is now recognized to be
one of the most important tasks in theoretical physics. Unification, in fact,
has been a feature of all great theories of physics. It is a well knonw fact
that Newton, Maxwell and Einstein, they all succeeded in performing some sort
of unification. So, not surprisingly in the last two centuries physicists have
recurrently pursued this theme. Broadly speaking one can mention two different
paths followed by theoreticians to arrive at a unified field theory.\ First,
there are the early attempts of Einstein, Weyl, Cartan, Eddington,
Schr\"{o}dinger and many others, whose aim consisted of unifying gravity and
electromagnetism \cite{Goenner}. Their approach consisted basically in
resorting to different kinds of non-Riemannian geometries capable of
accomodating new geometrical structures with a sufficient number of degrees of
freedom to describe the electromagnetic field. In this way, different types of
geometries have been invented, such as affine geometry, Weyl's geometry (where
the notion of parallel transport generalizes that of Levi-Civita's),
Riemann-Cartan geometry (in which torsion is introduced), to quote only a few.
In fact, it is not easy to track all further developments of these geometries,
most of which were clearly motivated by the desire of extending general
relativity to accomodate in its scope the electromagnetic field. However, as
we now see, the main problem with all these attempts was that they completely
ignored the other two fundamental interactions and did not take into account
quantum mechanics, dealing with unification only at classical level. Of
course, an approach to unification today would necessarily take into account
quantum field theory.

The second approach to unification has to do with the rather old idea that our
space-time may have more than four dimensions. This program starts with the
work of the Finnish physicist Gunnar Nordstr\"{o}m \cite{Nordström}, in 1914.
Nordstr\"{o}m realised that by postulating the existence of a fifth dimension
he was able (in the context of his scalar theory of gravitation) to unify
gravity and electromagnetism. Although the idea was quite original and
interesting, it seems the paper did not attract much attention due to the fact
that his theory of gravitation was not accepted at the time. Then, soon after
the completion of general relativity, Th\'{e}odor Kaluza, and later, Oscar
Klein, launched again the same idea, now entirely based on Einstein's theory
of gravity. Kaluza-Klein theory \ starts from five-dimensional vacuum
Einstein's equations and shows that, under certain assumptions, the field
equations reduce to a four-dimensional system of coupled Einstein-Maxwell
equations. This seminal idea has given rise to several different theoretical
developments, all of them exploring the possibility of achieving unification
from extra dimensionality of space-time. Indeed, through the old and modern
versions of Kaluza-Klein theory \cite{kaluza,kaluza1,kaluza2}, supergravity
\cite{deser}, superstrings \cite{superstrings}, and to the more recent
braneworld scenario\cite{rs1,rs2}, induced-matter \cite{overduin,book} and
M-theory \cite{duff}, there has been a strong belief among some physicists
that unification might finally be achieved if one is willing to accept that
space-time has more than four dimensions.

Among all these higher-dimensional theories, one of them, the induced-matter
theory (also referred to as space-time-matter theory \cite{overduin,book}) has
called our attention since it vividly recalls Einstein's belief that matter
and radiation (not only the gravitational field) should ultimately be viewed
as manifestations of pure geometry\cite{Einstein3}. Kaluza-Klein theory was a
first step in this direction. But it was Paul Wesson \cite{book}, from the
University of Waterloo, who pursued the matter further. Wesson and
collaborators realized that by embedding the ordinary space-time into a
five-dimensional vacuum space, it was possible to describe the macroscopic
properties of matter in geometrical terms. In a series of interesting papers
Wesson and his group showed how to produce standard cosmological models from
five-dimensional vacuum space. It looked like as if any energy-momentum tensor
could be generated by an embedding mechanism. At the time these facts were
discovered, there was no guarantee that \textit{any }energy-momentum tensor
could be obtained in this way. Putting it in mathematical terms, Wesson's
programm would not always work unless one could prove that \textit{any
}solution of Einstein's field equations could be isometrically embedded in
five-dimensional Ricci-flat space \cite{Romero}. It turns out, however, that
this is exactly the content of a beautilful and powerful theorem of
differential geometry now known as the Campbell-Magaard theorem
\cite{Campbell}. This theorem, little known until recently, was proposed by
English mathematician John Campbell in 1926, and was given a complete proof in
1963 by Lorenz Magaard \cite{Magaard}. Campbell \cite{Campbell}, as many
geometers of his time, was interested in geometrical aspects of Einstein's
general relativity and his works \cite{Campbell2} were published a few years
before the classical Janet-Cartan \cite{janet,cartan} theorem on embeddings
was established. Compared to the Janet-Cartan theorem the nice thing about the
Campbell-Magaard's result is that the codimension of the embedding space is
drastically reduced: one needs only one extra-dimension, and that perfectly
fits the requirements of the induced-matter theory. Finally, let us note that
both theorems refer to local and analytical embeddings (the global version of
Janet-Cartan theorem was worked out by John Nash \cite{Nash}, in 1956, and
adapted for semi-Riemannian geometry by R. Greene \cite{Greene}, in
1970,\ while a discussion of global aspects of Cambell-Magaard has recently
appeared in the literature \cite{Kat}).

\section{\qquad Higher-dimensional space-times and Riemannian extensions of
the Campbell-Magaard theorem}

Besides the induced-matter proposal, there appeared at the turn of the XX
century some other physical models of the Universe, which soon attracted the
attention of theoreticians. These models have put forward the idea that the
space-time of our everyday perception may be viewed as a four-dimensional
hypersurface embedded not in a Ricci-flat space, but in a five-dimensional
Einstein space (referred to as \textit{the bulk}) \cite{rs1,rs2}. Spurred by
this proposal new research on the geometrical structure of the proposed models
started. It was conjectured \cite{anderson} and later proved that the
Campbell-Magaard theorem could be immediately generalized for embedding
Einstein spaces \cite{dahia1}. This was the first extension of the
Campbell-Magaard theorem and other extensions, still in the context of
Riemannian geometry, were to come. More general local isometric embeddings
were next investigated, and it was proved that any $n$-dimensional
semi-Riemannian analytic manifold can be locally embedded in \ $(n+1)$%
-dimensional analytic manifold with a non-degenerate Ricci-tensor, which is
equal, up to a local analytical diffeomorphism, to the Ricci-tensor \ of an
arbitrary specified space \cite{dahia3}. Further motivation in this direction
came from studying embeddings in the context of non-linear sigma models, a
theory proposed by J. Schwinger in the fifties to describe strongly
interacting massive particles \cite{Schwinger}. It was then showed that any
$n$-dimensional Lorentzian manifold $(n\geq3)$ can be harmonically embedded in
an $(n+1)$-dimensional semi-Riemannian Ricci-flat manifold \cite{dahia4}.

At this point we should remark that most theories that regard our spacetime as
a hypersurface embedded in a higher-dimensional manifold \cite{Roy} make the
tacit assumption that this hypersurface has a semi-Riemannian geometrical
structure. Surely, this assumption avoids possible conflicts with the
well-established theory of general relativity, which operates in a Riemannian
geometrical frame. However, recently there has been some attempts to broaden
this scenario. For instance, new theoretical schemes have been proposed, where
one of the most simple generalizations of non-Riemannian geometry, namely the
Weyl geometry \cite{Weyl}, has been taken into consideration as a viable
possibility to describe the geometry of the bulk
\cite{Israelit,Arias,Nandinii}. In some of these approaches, the
induced-matter theory is revisited to show that\ it is even possible to
generate a cosmological constant, or rather, a cosmological function, from the
extra dimensions and the Weyl field \cite{madriz}. In a similar context, it
has also been shown how the presence of the Weyl field may affect both the
confinement and/or stablity of particles motion, and how a purely geometrical
field, such as the Weyl field, may effectively act both as a classical and
quantum scalar field, which in some theoretical-field modes is the responsible
for the confinement of matter in the brane \cite{jansen,gomez}.

There is also another very interesting and compelling argument for considering
a Weyl structure as a suitable mathematical model for describing space-time.
This is based on the well-known axiomatic approach to space-time theory put
forward by Ehlers, Pirani and Schild (EPS), which, through an elegant and
powerful theoretical construction, shows that by starting from a minimum set
of rather plausible and general axioms concerning the motion of light signals
and freely falling particles, one is naturally led to a Weyl structure as the
proper framework of space-time \cite{Pirani}. In order to reduce this more
general framework to that of a semi-Riemannian manifold we need an aditional
axiom to be added to this minimum set. \ It turns out, however, that this
added axiom does not seem as natural as the others, as was pointed out by
Perlick \cite{Perlik}. We take Perlick's point of view as one of the
motivations for investigating the geometry of Weyl spaces.

In this paper we shall consider the mathematical problem of extending
different versions of the Campbell-Magaard theorem from the Riemannian context
to Weyl's geometry. Specifically, we shall first analyze the possibility of
locally and analytically embedding an $n$-dimensional Weyl manifold in an
$(n+1)$-dimensional Weyl space, the latter being Ricci-flat. We then weaken
this condition to investigate the problem of embedding manifolds whose
symmetric part of the Ricci tensor vanishes. These problems can be regarded as
extensions of the Campbell-Magaard theorems, which hold in Riemannian
geometry, to a more general geometrical setting, namely that of Weyl's
geometry. We believe that an investigation of these seemingly purely
geometrical problems may also shed some light on the physics of
higher-dimensional theories in which there are extra degrees of freedom coming
from the geometric structure of space-time, in particular, those in which
there are mechanisms for generating matter and fields from extra dimensions in
the case of theories of gravitation whose geometrical framework is based on
the Weyl theory, and other higher-dimensional proposals formulated in Weyl
manifolds, such as D-dimensional dilaton gravity \cite{Bronnikov},
higher-dimensional WIST theories \cite{Israelit},\cite{Arias},\cite{madriz}%
,\cite{Novello} and others.

Finally, a few words should be said with regard to the Campbell-Maggard
theorem and its application to physics. First, let us note that the proof
provided by Magaard is based on the Cauchy-Kovalevskaya theorem. Therefore,
some properties of relevance to physics, such as the stability of the
embedding, cannot be guaranteed to hold \cite{Anderson}. Nevertheless, the
problem of embedding space-time into five-dimensional spaces can be considered
in the context of the Cauchy problem for general relativity \cite{Yvonne}.
Specifically, it has recently been shown that the embedded space-time may
arise as a result of physical evolution of proper initial data. This new
perspective has some advantages in comparison with the original
Campbell-Magaard formulation because, by exploring the hyperbolic character of
the field equations, it allows to show that the embedding has stability and
domain of dependence (causality) properties \cite{dahia5}.

\section{Weyl geometry}

When working in Riemannian geometry we consider a pair $(M,g)$, where $M$ is a
differentiable manifold and $g$ a (semi)-Riemannian metric defined on $M$. The
fundamental theorem of Riemannian geometry states that there is a unique
torsionless linear connection \textit{compatible} with $g$ \cite{do Carmo}.
By\textit{ compatibility} we mean the following. When we endow any
differentiable manifold with a linear connection $\nabla$ we have an
associated notion of parallel transport. It is well known that parallel
transport defines isomorphisms between tangent spaces. The compatibility
condition is defined as the requirement that this isomorphism is also an
isometry. This turns out to be equivalent to the following requirement
\[
\nabla g=0.
\]
It turns out that in Weyl's geometry we relax the requirement of $\nabla$
being compatible with $g$, and this means that parallel transports are not
required to define isometries anymore. We first endow $M$ not only with a
semi-Riemannian metric, but also with a one-form field $\omega$, so that
instead of the pair $(M,g)$ we now consider a triple $(M,g,\omega)$. Weyl's
connection is defined by requiring it to be torsionless and that, for any
parallel vector field $V$ along any smooth curve $\gamma$, the following
condition is satisfied:
\begin{equation}
\frac{d}{dt}g(V(t),V(t))=\omega(\gamma^{\prime}(t))g(V(t),V(t)).
\label{weyltransp}%
\end{equation}
Before presenting the main existence and uniqueness theorems for such
connection, we shall try to get some insight on what this condition means
geometrically. First of all, note that because parallel transport is a linear
application,\ if $V,W$ are parallel fields along some curve $\gamma$, then
$V+W$ also is a parallel field along $\gamma$. On the other hand, by
polarization we get
\[
g(V,W)=\frac{1}{2}(g(V+W,V+W)-g(V,V)-g(W,W)),
\]
which together with (\ref{weyltransp}) gives
\[
\frac{d}{dt}g(V,W)=\omega(\gamma^{\prime})g(V,W).
\]
We thus say that the connection $\nabla$ is \textit{Weyl compatible} with
$(M,g,\omega)$ if for any pair of parallel vectors along any smooth curve
$\gamma=\gamma(t)$ the condition below is satisfied
\[
\frac{d}{dt}g(V(t),W(t))=\omega(\gamma^{\prime}(t))g(V(t),W(t)),
\]
where $\gamma^{\prime}$ denotes the tangent vector of $\gamma.$ Integrating
the above equation along the curve $\gamma$ leads to
\begin{equation}
g(V(t),W(t))=g(V(0),W(0))e^{\int_{0}^{t}\omega(\gamma^{\prime}(s))ds}.
\label{weyltransp2}%
\end{equation}
In particular, if $V=W$ this last expression gives us precisely how much the
parallel transport fails to be an isometry:
\[
g(V(t),V(t))=g(V(0),V(0))e^{\int_{0}^{t}\omega(\gamma^{\prime}(s))ds}.
\]
Note that if the vectors $V(0)$ and $W(0)$ are orthogonal, then
(\ref{weyltransp2}) implies that they remain orthogonal when parallel
transported along the curve, although their respective "norms" may change.

Let us now state some results that hold for a Weyl connection which are
analogues to those valid for a Riemannian connection. All these results are
proven in very much the same way as in Riemannian geometry.

\begin{prop}
A connection $\nabla$ is compatible with a Weyl structure $(M,g,\omega)$ iff
for any pair of vector fields $V,W$ along any smooth curve $\gamma$ in $M$ the
following holds:
\begin{equation}
\frac{d}{dt}g(V,W)=g(\frac{DV}{dt},W)+g(V,\frac{DW}{dt})+\omega(\gamma
^{\prime})g(V,W) \label{weyltransp3}%
\end{equation}

\end{prop}

\begin{coro}
A linear connection $\nabla$ is compatible with a Weyl structure
$(M,g,\omega)$ iff $\forall$ $p\in M$ and for every vector fields $X,Y,Z$ on
$M$ the condition below holds
\begin{equation}
X_{p}(g(Y,Z))=g_{p}(\nabla_{X_{p}}Y,Z_{p})+g_{p}(Y_{p},\nabla_{X_{p}}%
Z)+\omega_{p}(X_{p})g_{p}(Y_{p},Z_{p}) \label{weyltransp4}%
\end{equation}

\end{coro}

In the last proposition we can actually drop $p$ and write
\[
X(g(Y,Z))=g(\nabla_{X}Y,Z)+g(Y,\nabla_{X}Z)+\omega(X)g(Y,Z),
\]
which, then, can be used to prove the following:

\begin{prop}
A linear connection $\nabla$ is compatible with a Weyl structure
$(M,g,\omega)$ iff it satisfies
\begin{equation}
\label{weylcompatibility}\nabla g=\omega\otimes g
\end{equation}

\end{prop}

Now the following result is easily established.

\begin{prop}
\label{Weyluniq} There is a unique torsionless connection compatible with the
Weyl structure $(M,g,\omega)$.
\end{prop}

In the proof of this proposition it is found that the Weyl connection, in a
particular coordinate system, takes the following form:
\begin{equation}
\label{weylconnec}\Gamma^{u}_{ac}=\frac{1}{2}g^{bu}(\partial_{a}%
g_{bc}+\partial_{c}g_{ab}-\partial_{b}g_{ca})+\frac{1}{2}g^{bu}(\omega
_{b}g_{ca}-\omega_{a}g_{bc}-\omega_{c}g_{ab})
\end{equation}

It is important to note that a Weyl manifold defines an equivalence class of
such structures all linked by the following group of transformations:
\begin{equation}
\left\{
\begin{array}
[c]{ll}%
\overline{g}=e^{-f}g & \\
\overline{\omega}=\omega-df &
\end{array}
\right.  \label{weyl group}%
\end{equation}
where $f$ is an arbitrary smooth function defined on $M$. It is easy to check
that these transformations define an equivalence relation between Weyl
manifolds, and that if $\nabla$ is compatible with $(M,g,\omega),$ then it is
also compatible with $(M,\overline{g},\overline{\omega})$. In this way every
member of the class is compatible with the same connection, hence has the same
geodesics, curvature tensor and any other property that depends only on the
connection. This is the reason why it is regarded more natural, when dealing
with Weyl manifolds, to consider the whole class of equivalence
$(M,[g],[\omega])$ rather than working with a particular element of this
class. In this sense, it is argued that only geometrical quantities that are
invariant under (\ref{weyl group}) are of real significance in the case of
Weyl geometry. Following the same line of argument it is assumed that only
physical theories and physical quantities presenting this kind of invariance
should be considered of interest in this context. To conclude this section, we
remark that when the one-form field $\omega$ is an exact form, then the Weyl
structure is called \textit{integrable}.

\subsection{Weyl submanifolds}

\begin{defn}
Let $(\overline{M},\overline{g},\overline{\omega})$ be a Weyl manifold and
$M\hookrightarrow\overline{M}$ be a submanifold of $\overline{M}$. If the
pullback $i^{*}(g)$ is a metric tensor on $M$ then $(M,i^{*}(\overline
{g}),i^{*}(\overline{\omega}))$ is a \textit{Weyl submanifold} of
$\overline{M}$. In this case we will use the notation $g=i^{*}(\overline{g})$
and $\omega=i^{*}(\overline{\omega})$ for the induced metric and 1-form.
\end{defn}

Using the same conventions as in the previous definition, we denote by
$\overline{\nabla}$ the Weyl-compatible connection associated with
$(\overline{M},\overline{g},\overline{\omega})$. We define the induced
connection $\nabla$ on $M$ following the same reasoning as in the Riemannian
case. Thus if $X,Y$ are vector fields on $M$, and $\overline{X},\overline{Y}$
are extensions of these vector fields to $\overline{M}$, then $\nabla
_{X}Y\doteq(\overline{\nabla}_{\overline{X}}\overline{Y})^{T}$. It is a
well-known fact that this definition does not depend on the extensions
\cite{do Carmo}.

It is worth noticing that both the definition of Weyl submanifold and of
induced connection make sense in the whole class $(\overline{M},[\overline
{g}],[\overline{\omega}])$. We can see that the definition of Weyl submanifold
satisfies this condition since every such structure that can be obtained from
an element of $(\overline{M},[\overline{g}],[\overline{\omega}])$ lies in
$(M,[g],[\omega])$ and vice versa, every element of $(M,[g],[\omega])$ can be
obtained from some element of $(\overline{M},[\overline{g}],[\overline{\omega
}])$. The fact that the definition of induced connection is invariant in the
whole class $(\overline{M},[\overline{g}],[\overline{\omega}])$, is because
two conformal metrics make the same splitting of the tangent spaces:
$T_{p}\overline{M}=T_{p}M\oplus T_{p}M^{\perp}$.

The following results are obtained in the same way as in Riemannian geometry.

\begin{prop}
Given a Weyl manifold $(\overline{M},\overline{g},\overline{\omega})$ and a
Weyl submanifold $M\hookrightarrow\overline{M}$, the induced connection
$\nabla$ on $M$ is the Weyl connection compatible with the induced Weyl
structure on $(M,g,\omega)$.
\end{prop}

As usual, we define the second fundamental form $\alpha$ on $M$ as
\begin{align*}
\alpha:TM\times TM  & \mapsto TM^{\perp}\\
\alpha(X,Y)  &  \doteq(\overline{\nabla}_{\overline{X}}\overline{Y})^{\perp}%
\end{align*}
One can easily check that this definition does not depend on how we extend $X$
and $Y$ to $\overline{M}$. Thus, if $X,Y$ are fields on $M$ we write
\[
\overline{\nabla}_{X}Y=\nabla_{X}Y+\alpha(X,Y)
\]
The next proposition is analogous to its Riemannian counterpart:

\begin{prop}
The second fundamental form $\alpha$ is symmetric and $\mathfrak{F}(M)$-linear
in both arguments.
\end{prop}

From now on we shall consider only hypersurfaces. In this case, we can define
a unit normal vector field $\eta$, which, at least locally, is unique up to a
sign. We define the \textit{scalar second fundamental form }$l$ as given by%

\begin{align*}
l:TM\times TM  &  \mapsto\mathfrak{F}(M)\\
(X,Y)  &  \mapsto\overline{g}(\alpha(X,Y),\eta)
\end{align*}
We note that although the choice of the unit normal field $\eta$ depends on a
particular element of $(\overline{M},[\overline{g}],[\overline{\omega}])$, the
definition of $l$ does not.

Now from the last proposition it follows that $l$ is symmetric and
$\mathfrak{F}(M)$-linear, i.e., $l$ is a symmetric $(0,2)$-tensor field on
$M$. \ Following a procedure entirely analogous to what is done in Riemannian
geometry, we obtain the Gauss-Codazzi equations for hypersurfaces. Thus, \ if
$X$,$Y$ and $Z$ are vector fileds on $M,$ Gauss' equation takes the form
\begin{equation}
\overline{g}(\overline{R}(X,Y)Z,W)=\overline{g}(R(X,Y)Z,W)+\overline{g}%
(\alpha(X,Z),\alpha(Y,W))-\overline{g}(\alpha(Y,Z),\alpha(X,W)).
\label{Gausseq}%
\end{equation}

If $\xi$ is a unit field normal to $M$, then Codazzi's equation reads
\begin{equation}
\overline{g}(\overline{R}(X,Y)Z,\xi)=\epsilon((\nabla_{X}l)(Y,Z)-(\nabla
_{Y}l)(X,Z)+\frac{1}{2}(\omega(Y)l(X,Z)-\omega(X)l(Y,Z))). \label{codazzieq2}%
\end{equation}
Where $\epsilon\doteq\overline{g}(\xi,\xi)=\pm1$ and the sign depends on
whether the restriction of $\overline{g}$ to each $T_{p}M^{\perp}$ is positive
or negative definite.

Let us now have a look at the Bianchi identities in Weyl geometry, as they
will be useful in our investigation of the embedding problem.

\subsection{Bianchi identities}

We start by writing down the symmetries of the Riemann curvature tensor
$\mathcal{R}$ defined on an $n$-dimensional Weyl manifold. First of all, in
order to clarify notation, we remark that in this paper we adopt the following
convention for the curvature tensor:
\[
R^{\rho}{}_{\sigma\beta\alpha}=\partial_{\alpha}\Gamma_{\beta\sigma}^{\rho
}-\partial_{\beta}\Gamma_{\alpha\sigma}^{\rho}+\Gamma_{\beta\sigma}^{\gamma
}\Gamma_{\alpha\gamma}^{\rho}-\Gamma_{\alpha\sigma}^{\gamma}\Gamma
_{\beta\gamma}^{\rho}.
\]

In terms of the components of $\mathcal{R}$ in a coordinate basis it is easy
to see that for any connection the following identity holds:
\[
R^{\rho}{}_{\mu\nu\alpha}=-R^{\rho}{}_{\mu\alpha\nu}.
\]
Moreover, if the connection is torsionless we also have the Bianchi
identities
\begin{align}
R^{\rho}{}_{\mu\nu\alpha}+R^{\rho}{}_{\alpha\mu\nu}+R^{\rho}{}_{\nu\alpha\mu}
&  =0,\label{biachi1}\\
R^{\rho}{}_{\mu\nu\alpha;\lambda}+R^{\rho}{}_{\mu\lambda\nu;\alpha}+R^{\rho}%
{}_{\mu\alpha\lambda;\nu}  &  =0, \label{bianchi2}%
\end{align}
where the semicolon denotes covariant differentiation.

We now look for a contracted version of the Bianchi identities. In particular,
we want to get a geometric identity for $g^{\alpha\beta}\nabla_{\alpha}%
^{S}G_{\beta\sigma}$, where the upper index $S$ stands for the "symmetric
part". Before doing this we need one more identity, which comes from looking
at the following expression for the Riemann tensor
\begin{equation}
R^{\rho}{}_{\sigma\beta\alpha}={}^{\circ}R^{\rho}{}_{\sigma\beta\alpha
}+g_{\sigma\lbrack\beta}{}^{\circ}\nabla_{\alpha]}\omega^{\rho}-\delta
_{\sigma}^{\rho}{}^{\circ}\nabla_{\lbrack\alpha}\omega_{\beta]}-\delta
_{\lbrack\beta}^{\rho}{}^{\circ}\nabla_{\alpha]}\omega_{\sigma}+\frac{1}%
{2}\delta_{\lbrack\alpha}^{\rho}\omega_{\beta]}\omega_{\sigma}-\frac{1}%
{2}g_{\sigma\lbrack\beta}\delta_{\alpha]}^{\rho}\omega^{\gamma}\omega_{\gamma
}-\frac{1}{2}g_{\sigma\lbrack\alpha}\omega_{\beta]}\omega^{\rho},
\label{riemannweyl}%
\end{equation}
where ${}^{\circ}$ denotes quantities computed with the Riemannian connection.
From this expression we can prove the identity
\[
R_{\lambda\sigma\beta\alpha}+R_{\sigma\lambda\beta\alpha}=2g_{\lambda\sigma
}F_{\beta\alpha},
\]
where $F_{\beta\alpha}=d\omega_{\alpha\beta}=\frac{1}{2}(\nabla_{\beta}%
\omega_{\alpha}-\nabla_{\alpha}\omega_{\beta})$. In order to compute
$g^{\alpha\beta}\nabla_{\alpha}^{S}G_{\beta\sigma}$ we can first compute both
divergences, which will give us the final result. Using all the previous
identities, it is not difficult to see that we are led to the following:
\begin{equation}
g^{\mu\lambda}\nabla_{\lambda}^{S}G_{\nu\mu}=\frac{n-2}{2}g^{\mu\lambda}%
\nabla_{\lambda}F_{\nu\mu}. \label{contractedbianchi}%
\end{equation}

\section{The embedding problem}

We now turn to the problem of existence of isometric embeddings of Weyl
manifolds. In particular, we are interested in studying possible extensions of
the Campell-Magaard-like theorems in the context of Weyl geometry. First of
all, we shall define what we understand by an\textit{ isometric embedding } in
this context.

\begin{defn}
An \textit{isometric immersion} $\phi:M\mapsto\tilde{M}$ between two Weyl
manifolds $(M,g,\omega)$ and $(\tilde{M},\tilde{g},\tilde{\omega})$ is a
smooth mapping satisfying:
\begin{align*}
1)  &  d\phi_{p}\;\;is\;\;injective\;\;\forall p\in M\\
2)  &  \phi^{\ast}(\tilde{g})=g\\
3)  &  \phi^{\ast}(\tilde{\omega})=\omega.
\end{align*}
If, furthermore, $\phi$ is one-to-one and the induced map $M\mapsto\phi(M)$ is
an homeomorphism, where $\phi(M)\subset\tilde{M}$ is seen with the induced
topology, then we say that $\phi$ is an \textit{isometric embedding}. Also, we
shall say that $\phi$ is a \textit{local isometric embedding at} $p\in M$ if
there is a neighborhood of $p$ where $\phi$ is an embedding .
\end{defn}

An important result we shall use when studying Campbell-Magaard-like theorems
is the following theorem:

\begin{thm}
\label{thm1} Let $(M^{n},g,\omega)$ and $(\tilde{M}^{n+1},\tilde{g}%
,\tilde{\omega})$ be Weyl manifolds, and $(U,\mu)$ a coordinate system around
$p\in M^{n}$. Then $(M,g,\omega)$ has a local analytic isometric embedding in
$(\tilde{M}^{n+1},\tilde{g},\tilde{\omega})$ around $p$ iff there are analytic
functions $\overline{g}_{ik}$, $\overline{\psi}$, $\overline{\omega}_{k}$ and
$\tilde{\omega}_{n+1}$, with $i,k=1,\dots,n$, defined on some open set
$D\subset\mu(U)\times\mathbb{R}$ containing $(x_{p}^{1},\dots,x_{p}^{n},0)$
satisfying the following conditions
\[
\overline{g}_{ik}(x^{1},\dots,x^{n},0)=g_{ik}(x^{1},\dots,x^{n})
\]%
\[
\overline{\omega}_{k}(x^{1},\dots,x^{n},0)=\omega_{k}(x^{1},\dots,x^{n})
\]
on some open set $A\subset\mu(U)$, and
\begin{equation}
\overline{g}_{ik}=\overline{g}_{ki} \label{immercond2}%
\end{equation}%
\begin{equation}
\det(\overline{g}_{ik})\neq0 \label{immercond3}%
\end{equation}%
\begin{equation}
\overline{\psi}\neq0 \label{immercond4}%
\end{equation}
on $D$, and such that on some open set $V\subset\tilde{M}^{n+1}$, the metric
$\tilde{g}$ and the $1-$form $\tilde{\omega}$ can be written in coordinates
as
\[
\tilde{g}=\overline{g}_{ik}dx^{i}\otimes dx^{k}+\epsilon\overline{\psi}%
^{2}dx^{n+1}\otimes dx^{n+1}%
\]%
\[
\tilde{\omega}=\overline{\omega}_{k}dx^{k}+\tilde{\omega}_{n+1}dx^{n+1}%
\]

where $\epsilon=\pm1$.
\end{thm}

At first sight a natural extension of the Campbell-Magaard theorem
\cite{dahia1} in the context of Weyl geometry seems to be to prove the
existence of a local analytic isometric embedding of an arbitrary Weyl
manifold $(M^{n},g,\omega)$ in an $(n+1)$-dimensional Weyl manifold
$(\tilde{M}^{n+1},\tilde{g},\tilde{\omega})$ satisfying $\tilde{R}%
_{\alpha\beta}=0$ around some arbitrary point $p\in M$. This turns out to be a
simple extension which can be treated in complete analogy to \cite{dahia1}
after making some considerations. First, note that $\tilde{R}_{\alpha\beta}=0$
implies that both its symmetric and antisymmetric parts of $\tilde{R}%
_{\alpha\beta}$ \ must vanish. However, we already know that for an
$n-$dimensional Weyl manifold ${}$\ we have ${}^{A}\tilde{R}_{\alpha\beta
}=\frac{n}{2}F_{\alpha\beta}$. Therefore, this condition implies
$F_{\alpha\beta}=0$, which is locally equivalent to setting $\omega=d\phi$,
for some function $\phi$. In other words, in this case $(\tilde{M},\tilde
{g},\tilde{\omega})$ gives an integrable Weyl structue. From this, we see that
if $(M,g,\omega)$ is non-integrable, then it does not exist any isometric
embedding of $(M,g,\omega)$ in Ricci-flat manifolds, irrespective of the
codimension considered. Thus, let us first consider integrable Weyl manifolds
and look for analytic isometric embeddings of an integrable Weyl manifold
$(M^{n},g,\phi)$ in a Ricci-flat integrable Weyl manifold $(\tilde{M}%
^{n+1},\tilde{g},\tilde{\phi})$. We now proceed to set up the notation that
will be used throughout this paper.

Henceforth we shall consider $\tilde{M}=M\times\mathbb{R}$, a local chart in
$M$ defined in a neighbourhood of $p$ with coordinates $(x^{1},\dots,x^{n})$,
while in the product structure we have a coordinate system around $(p,0)$ with
coordinates $(x^{1},\dots,x^{n},y)$, where $y$ denotes the coordinate in
$\mathbb{R}$. In this coordinate system we write%

\[%
\begin{split}
\tilde{g}  &  =\overline{g}_{ik}dx^{i}\otimes dx^{k}+\epsilon\overline{\psi
}^{2}dy\otimes dy,\\
\tilde{\omega}  &  =\overline{\omega}_{i}dx^{i}+\tilde{\omega}_{n+1}dy,
\end{split}
\]
and consider the unit normal field given by
\[
\xi=\frac{1}{|\tilde{g}(\partial_{n+1},\partial_{n+1})|^{\frac{1}{2}}}%
\partial_{n+1}.
\]
From Gauss' equation we obtain
\begin{equation}
\tilde{R}_{likj}=\overline{R}_{likj}+\epsilon(l_{ji}l_{kl}-l_{ki}l_{jl}),
\label{riccigauss}%
\end{equation}
with
\begin{equation}
l_{ji}=\epsilon\overline{\psi}\tilde{\Gamma}_{ji}^{n+1}=-\frac{1}%
{2\overline{\psi}}\frac{\partial}{\partial y}\overline{g}_{ij}+\frac
{1}{2\overline{\psi}}\overline{g}_{ij}\tilde{\omega}_{n+1}.
\label{secondform1}%
\end{equation}
Also, from the Gauss-Codazzi equations and some explicit expressions for the
components of the connection, we arrive at the following equations for the
components of the Ricci tensor:
\begin{align}
\label{riccicomponents}%
\begin{split}
\tilde{R}_{ij}  &  =\overline{R}_{ij}+\epsilon\overline{g}^{kl}(l_{ij}%
l_{kl}-2l_{ki}l_{jl})+\frac{1}{\overline{\psi}}\overline{\nabla}_{j}%
\overline{\nabla}_{i}\overline{\psi}-\frac{\epsilon}{\overline{\psi}}%
\partial_{n+1}l_{ij}-\frac{1}{2}\overline{\nabla}_{j}\overline{\omega}_{i}\\
&  +\frac{1}{4}\overline{\omega}_{i}\overline{\omega}_{j}-\frac{1}%
{2\overline{\psi}}(\overline{\omega}_{i}\partial_{j}\overline{\psi}%
+\overline{\omega}_{j}\partial_{i}\overline{\psi}-\epsilon\tilde{\omega}%
_{n+1}l_{ji})\\
\tilde{R}_{(n+1)i}  &  =\epsilon\overline{\psi}\overline{g}^{kl}%
(\overline{\nabla}_{k}l_{il}-\overline{\nabla}_{i}l_{kl})+\epsilon
\frac{\overline{\psi}}{2}\overline{g}^{kl}(\overline{\omega}_{i}%
l_{kl}-\overline{\omega}_{k}l_{il})-\frac{1}{2}(\partial_{i}\tilde{\omega
}_{n+1}-\partial_{n+1}\overline{\omega}_{i})\\
\tilde{R}_{i(n+1)}  &  =\epsilon\overline{\psi}\overline{g}^{kl}%
(\overline{\nabla}_{k}l_{il}-\overline{\nabla}_{i}l_{kl})+\epsilon
\frac{\overline{\psi}}{2}\overline{g}^{kl}(\overline{\omega}_{i}%
l_{kl}-\overline{\omega}_{k}l_{il})+\frac{n}{2}(\partial_{i}\tilde{\omega
}_{n+1}-\partial_{n+1}\overline{\omega}_{i})\\
\tilde{R}_{(n+1)(n+1)}  &  =-\overline{\psi}^{2}\overline{g}^{jk}\overline
{g}^{iu}l_{uk}l_{ji}+\frac{1}{2}\overline{\psi}\tilde{\omega}_{n+1}l_{i}%
^{i}-\overline{\psi}\overline{g}^{iu}\partial_{n+1}l_{ui}+\epsilon
\overline{\psi}\overline{\nabla}_{i}\overline{\nabla}^{i}\overline{\psi}-\\
&  \frac{\epsilon}{2}\overline{\psi}^{2}\overline{\nabla}_{i}\overline{\omega
}^{i}-\frac{\epsilon}{4}\overline{\psi}^{2}\overline{\omega}^{i}%
\overline{\omega}_{i}%
\end{split}
\end{align}
Our next step is to compute the component $\tilde{G}_{n+1}^{n+1}$ of the
Einstein tensor to obtain
\begin{equation}
\tilde{G}_{n+1}^{n+1}=-\frac{1}{2}(\overline{R}+\epsilon\overline{g}%
^{ij}\overline{g}^{kl}(l_{ij}l_{kl}-l_{ki}l_{jl})). \label{G^{n+1}_{n+1}}%
\end{equation}

\subsection{The Weyl integrable case}

In this section we will discuss the embedding problem for Weyl integrable
manifolds. It is worth noticing that, in this case, there is a stronger
analogy with some Riemannan problems already studied in contact with General
Relativity. This is because if, for one particular member of the class
$(\tilde{M}^{n+1},[\tilde{g}],[\tilde{\phi}])$, we split the Ricci tensor into
its \textit{Riemannian part} and the \textit{extra terms}, then the Ricci-flat
condition becomes equivalent to the Einstein field equations with a scalar
field as a source. In the Riemannian framework, embeddings in such structures
have been studied by Ponce de Leon, who constructed explicit embeddings of
general vacuum solutions of $n$-dimensional general relativity (with a
possible presence of the cosmological constant) into $(n+1)$-Semi-Riemannian
manifolds sourced by a scalar field \cite{Ponce de Leon}. We should also
mention that embeddings in such structures where also treated by Anderson
\textit{et al.}, in which they worked out one of the known extensions of the
Campell-Magaard theorem \cite{CM-Scalar}. Even though these results are
clearly related to the problem we intend to study here, there are important
differences, one of them and maybe the main one, is that, since in both
\cite{Ponce de Leon} and \cite{CM-Scalar} the underlying structure is
Riemannian, the results presented there would not guarantee the embedding of a
whole Weyl integrable structure $(M^{n},[g],[\phi])$ in a Ricci-flat Weyl
integrable structure $(\tilde{M}^{n+1},[\tilde{g}],[\tilde{\phi}])$, as will
be shown in this section. Another difference with respect to \cite{Ponce de
Leon} is that there it is shown that, given a solution of the vacuum Einstein
field equations in $n$-dimensions, then it is possible to construct embeddings
for such a solutions in $(n+1)$-dimensional manifolds sourced by scalar
fields. In contrast, we will not impose any restriction, besides the
regularity assumptions, for the initial data (it does not need to solve any
field equations on the original manifold). In this way, we can make an
interesting contact with these results known in Riemannian geometry, while
having some important differences with them.

In order to start with the discussion of the present embedding problem, note
that in the case where $(\tilde{M},\tilde{g},\tilde{\omega})$ is integrable,
that is, when $\tilde{\omega}=d\tilde{\phi}$, the expressions in
(\ref{riccicomponents}) are simplified. In fact, as we have already seen, the
Ricci tensor turns out to be symmetric in this case, and from
(\ref{contractedbianchi}) we obtain
\begin{equation}
\label{integrabledivergence}\tilde{g}^{\alpha\beta}\tilde{\nabla}_{\alpha
}\tilde{G}_{\nu\beta}=0.
\end{equation}

From the above, we see that the natural approach to the problem is to follow
the same procedure adopted in \cite{dahia1}, which consists in considering the
\textit{evolution} equations $\tilde{R}_{ij}=0$ in a neighborhood of
$0\in\mathbb{R}^{n+1}$, as well as the \textit{constraint} equations
$\tilde{R}_{i(n+1)}=0$ and $\tilde{G}_{n+1}^{n+1}=0$ on the hypersurface
$\Sigma_{0}$ given by $y=0$. Then, the evolution equations together with the
identity (\ref{integrabledivergence}), guarantee that we can propagate the
constraint equations in a neighborhood of the origin of $\mathbb{R}^{n+1}$. In
this scheme, we just consider $\tilde{\phi}$ as being some given analytic
function in a neighborhood of the origin satisfying $\tilde{\phi}%
(x,0)=\phi(x)$. Proceeding in this way, we find that the problem is totally
analogous to the one investigated in \cite{dahia1}, immediately leading to the
following statement:

\begin{thm}
Any analytic integrable n-dimensional Weyl manifold $(M^{n},g,\phi)$ admits a
local analytic isometric embedding around any point $p\in M$ in an analytic
Ricci-flat integrable Weyl manifold $(\tilde{M}^{n+1},\tilde{g},\tilde{\phi})$.
\end{thm}

It is interesting to note that this result guarantees the existence of
isometric embeddings for \textit{Weyl manifolds}, not for a \textit{Weyl
structure} $(M,[g],[\omega])$. Indeed, in order to take into account the whole
\textit{Weyl structure} we need to show that for every element of
$(M,[g],[\omega])$ there is an isometric embedding of this element in an
element of some $(n+1)-$dimensional Weyl structure $(\tilde{M},[\tilde
{g}],[\tilde{\omega}])$. Since, as already remarked, when working with Weyl
manifolds all the relevant geometric (and physical) quantities are to be
defined on the whole class, it is of much more interest to look for an
embedding for the whole structure. We claim that we can show this from our
previous results. To do this, let us consider the following argument.

Suppose that a particular $n-$dimensional Weyl manifold $(M,g,\omega)$ admits
a local analytic isometric embedding into an $(n+1)-$dimensional Weyl manifold
$(\tilde{M},\tilde{g},\tilde{\omega}),$ and that this embedding has been
constructed following our previous prescription, namely, that the embedding is
just the inclusion. On the other hand, any other element of the class
$(M,[g],[\omega])$ can be written as $(M,e^{-h}g,\omega-dh)$ for some analytic
function $h$. The question is whether there is some analytic function $f$ on
$\tilde{M}$ such that, for this element of the class, there is a local
analytic isometric embedding into $(\tilde{M},e^{-f}\tilde{g},\tilde{\omega
}-df)$. By using the same set up we have developed, we define the function
$f(x,y)$ in a neighborhood of the point $p\in\tilde{M}$ (where we know the
isometric embedding exists) by:
\[
f(x,y)\doteq h(x)+y.
\]
We then get
\[%
\begin{split}
e^{-f(x,0)}\tilde{g}_{ij}(x,0)  &  =e^{-h(x)}g_{ij}(x),\\
\tilde{\omega}_{i}(x,0)-\partial_{i}f(x,0)  &  =\omega_{i}(x)-\partial
_{i}h(x),
\end{split}
\]
which gives us the isometry condition. Also, since the Ricci tensor is an
invariant of the class of Weyl manifolds, we have shown the following result.

\begin{thm}
Any analytic $n$-dimensional integrable Weyl structure $(M^{n},[g],[\phi])$
admits a local analytic isometric embedding in an $(n+1)$-dimensional
integrable Weyl structure $(\tilde{M}^{n+1},[\tilde{g}],[\tilde{\phi}])$ with
vanishing Ricci tensor.
\end{thm}

We now turn our attention to the more general problem of embedding of Weyl
manifolds which are not necessarily integrable, dropping the condition of
Ricci-flatness. Thus, in the following sections, we shall weaken this latter condition.

\section{Embeddings in Weyl manifolds whose Ricci tensor has vanishing
symmetric part}

In this section we shall investigate the existence of a local isometric
embedding of an arbitrary Weyl manifold $(M^{n},g,\omega)$ around some point
$p\in M$ in a Weyl manifold $(\tilde{M}^{n+1},\tilde{g},\tilde{\omega})$ which
has $^{S}\tilde{R}_{\alpha\beta}=0$. This is the same as requiring that
$^{S}\tilde{G}_{\alpha\beta}=0$. From the identity
\[
\tilde{g}^{\alpha\beta}\tilde{\nabla}_{\alpha}^{S}\tilde{G}_{\nu\beta}%
=\frac{n-2}{2}\tilde{g}^{\alpha\beta}\tilde{\nabla}_{\alpha}F_{\nu\beta},
\]
we see that our requirement on $(\tilde{M},\tilde{g},\tilde{\omega})$ imposes
the condition
\begin{equation}
\tilde{g}^{\alpha\beta}\tilde{\nabla}_{\alpha}F_{\nu\beta}=0,
\label{condition1}%
\end{equation}
which must hold in a neighborhood of $p$. \ As we shall see, (\ref{condition1}%
)\ will impose further restrictions on $(\tilde{M},\tilde{g},\tilde{\omega})$.
To see this we shall need to make use of some geometric identities.

\begin{prop}
Suppose we have a semi-Riemannian manifold $M$ endowed with a torsionless
connection $\nabla$. Then, for any $T\in\mathfrak{X}_{2}^{0}(M)$ the following
identity holds:
\[
\nabla_{\nu}\nabla_{\mu}T_{\alpha\beta}-\nabla_{\mu}\nabla_{\nu}T_{\alpha
\beta}=-R^{\sigma}{}_{\alpha\mu\nu}T_{\sigma\beta}-R^{\sigma}{}_{\beta\mu\nu
}T_{\alpha\sigma}.
\]

\end{prop}

A corollary of this proposition in the context of Weyl geometry is given by
the statement below.

\begin{coro}
Suppose we have a Weyl manifold $(M,g,\omega)$, endowed with its
Weyl-compatible connection $\nabla$, and let $F\doteq d\omega$. Then, for any
2-form $T$ on $M$ we have the identity
\[
g^{\nu\alpha}g^{\mu\beta}\nabla_{\nu}\nabla_{\mu}T_{\alpha\beta}%
=-R^{\sigma\beta}T_{\sigma\beta}+2F^{\sigma\beta}T_{\sigma\beta}%
\]

\end{coro}

A direct consequence of the above is the following:

\begin{coro}
\label{coro2} Let $(M,g,\omega)$\ be a $n$-dimensional Weyl manifold whose
symmetric part of the Ricci tensor is zero. Then, for $n\neq4$ we must have
\begin{equation}
F^{\mu\nu}F_{\mu\nu}=0.
\end{equation}

\end{coro}

\begin{proof}
Using Weyl's compatibility condition we get the following:
\begin{align*}
g^{\mu\nu}\nabla_{\mu}(g^{\alpha\beta}\nabla_{\alpha}F_{\nu\beta})=g^{\mu\nu}g^{\alpha\beta}\nabla_{\mu}\nabla_{\alpha}F_{\nu\beta}-\omega^{\nu}g^{\alpha\beta}\nabla_{\alpha}F_{\nu\beta}
\end{align*}
We know that under our hypotheses (\ref{condition1}) is satisfied. Then the second term in the right-hand side of the previous expression vanishes and so does the left-hand side. Also we know that the previous corollary holds for the 2-form $F$. This gives us the following:
\begin{align*}
0&=-R^{\mu\nu}F_{\mu\nu}+2F^{\mu\nu}F_{\mu\nu}\\
&=-^{A}R^{\mu\nu}F_{\mu\nu}+2F^{\mu\nu}F_{\mu\nu}
\end{align*}
Using the fact that for a Weyl manifold of dimension $n$, the antisymmetric part of its Ricci tensor is $^{A}R_{\mu\nu}=\frac{n}{2}F_{\mu\nu}$ we get the following:
\begin{align*}
0=\frac{4-n}{2}F^{\mu\nu}F_{\mu\nu}
\end{align*}
So we get that if $n\neq 4$ then it must hold that:
\begin{align*}
F^{\mu\nu}F_{\mu\nu}=0
\end{align*}
\end{proof}

It is worth noticing that the above condition will lead to unexpected and
interesting \textit{no go} results. For example, if $\tilde{g}$ is a positive
definite metric, then $F^{\mu\nu}F_{\mu\nu}=0$ implies $F_{\mu\nu}=0$; hence
$(\tilde{M},\tilde{g},\tilde{\omega})$ is integrable. Therefore, we have the
following result:

\begin{thm}
Let $(M,g,\omega)$ be an n-dimensional non-integrable Weyl manifold, with
$n\geq5$. If $g$ is positive definite, then it is not possible to
isometrically immerse $(M,g,\omega)$ into a Weyl manifold $(\tilde{M}%
,\tilde{g},\tilde{\omega})$, with a positive definite metric $\tilde{g}$ and a
Ricci tensor, whose symmetric part is vanishing, regardless of the codimension
of the embedding .
\end{thm}

This result shows that the previous corollary imposes a very strong
restriction on the existence of embeddings in the case of Weyl manifolds. For
example, \textit{Theorem 4} implies that, rather surprisingly, for a
non-integrable Weyl manifold of dimension greater that 4, there does not exist
an isometric immersion in a Riemann-flat space.

We shall now treat the very particular 4-dimensional case for which this
restriction does not apply. In doing this we will make use of the restriction
on the dimensionality of the embedding manifold only when necessary, so that
the difficulties implied for the general dimensional case are made explicit.

\subsection{The 4-dimensional case}

The idea is to divide the equations $^{S}\tilde{R}_{\alpha\beta}=0$ into a of
set constraint equations and a set of evolution equations. To do this, we
shall impose an additional set of equations coming from the contracted Bianchi
identities. Explicitly, we shall impose the equations $\tilde{g}^{\alpha\beta
}\tilde{\nabla}_{\alpha}^{S}\tilde{G}_{\beta\sigma}=0$, which, as can be seen
from (\ref{contractedbianchi}), is equivalent to imposing the following set of
additional partial differential equations (PDE):
\begin{equation}
\tilde{g}^{\alpha\beta}\tilde{\nabla}_{\alpha}F_{\sigma\beta}=0. \label{edp4}%
\end{equation}
The above equations will be looked upon as a set of equations imposed on
$\tilde{\omega}_{\beta}$. Thus, our complete system consists of (\ref{edp4})
together with the following set of equations:
\begin{align}
^{S}\tilde{R}_{ij} & =0\label{edp1}\\
^{S}\tilde{R}_{i(n+1)} & =0\label{edp2}\\
^{S}\tilde{G}_{n+1}^{n+1} & =0\label{edp3}%
\end{align}
As we shall show, by using this scheme we can treat the problem as consisting
of a set of evolution equations plus some constraint equations.

\begin{lemma}
\label{lemma1} Let $\overline{g}_{ik}(x,y),\overline{\psi}(x,y)$ and
$\tilde{\omega}_{\alpha}(x,y)$ be analytic functions at $0\in\Sigma_{0}%
\subset\mathbb{R}^{n+1}$. Suppose that $\overline{g}_{ik}=\overline{g}_{ki}$,
$det(\overline{g}_{ik})\neq0$ and $\overline{\psi}\neq0$ in a neighborhood of
$0\in\mathbb{R}^{n+1}$, that $\overline{g}_{ik},\overline{\psi}$ and
$\tilde{\omega}_{\alpha}$ satisfy (\ref{edp4}) and (\ref{edp1}) in a
neighborhood $V$ of $0\in\mathbb{R}^{n+1}$ and also (\ref{edp2}) and
(\ref{edp3}) in a neighborhood of $0\in\Sigma_{0}$. Then, $\overline{g}%
_{ik},\overline{\psi}$ and $\tilde{\omega}_{\alpha}$ will satisfy (\ref{edp2})
and (\ref{edp3}) in a neighbourhood of $0\in\mathbb{R}^{n+1}$.
\end{lemma}

\begin{proof}
Since equation (\ref{edp4}) is equivalent to $\tilde{g}^{\alpha\beta}\tilde{\nabla}_{\alpha}^{S}\tilde{G}_{\beta\sigma}=0$, then by hypothesis we have that:
\begin{align*}
\overline{g}^{ij}\tilde{\nabla}_j^{S}\tilde{G}_{i\sigma}+\frac{\epsilon}{\overline{\psi}^2}\tilde{\nabla}_{n+1}{}^{S}\tilde{G}_{(n+1)\sigma}=0
\end{align*}
which is equivalent to the following:
\begin{align}\label{lemma1eq1}
\frac{\partial^{S}\tilde{G}_{(n+1)\sigma}}{\partial y}=-\epsilon\overline{\psi}^2\overline{g}^{ij}\partial_j^{S}\tilde{G}_{i\sigma}+\tilde{\Gamma}^{\gamma}_{(n+1)(n+1)}{}^{S}\tilde{G}_{\gamma\sigma}+\tilde{\Gamma}^{\gamma}_{(n+1)\sigma}{}^{S}\tilde{G}_{(n+1)\gamma}+\epsilon\overline{\psi}^2\overline{g}^{ij}(\tilde{\Gamma}^{\gamma}_{ij}{}^{S}\tilde{G}_{\gamma\sigma}+\tilde{\Gamma}^{\gamma}_{j\sigma}{}^{S}\tilde{G}_{i\gamma})
\end{align}
To analyze these equations firts set $\sigma=k$. We can use the fact that since (\ref{edp1}) holds in a neighborhood of $0\in\mathbb{R}^{n+1}$, then in such a neighborhood we have that the following holds $^{S}\tilde{G}_{ik}=-\frac{\epsilon}{\overline{\psi}^2}\tilde{g}_{ik}{}^{S}\tilde{G}_{(n+1)(n+1)}$. Then we get that:
\begin{align}\label{lemma1eq2}
\begin{split}
\frac{\partial}{\partial y}{}^{S}\tilde{G}_{(n+1)k}=&\partial_k^{S}\tilde{G}_{(n+1)(n+1)}+\overline{\psi}^2\overline{g}^{ij}\partial_j(\frac{\overline{g}_{ik}}{\overline{\psi}^2}){}^{S}\tilde{G}_{(n+1)(n+1)}-\frac{\epsilon}{\overline{\psi}^2}\tilde{\Gamma}^{j}_{(n+1)(n+1)}\overline{g}_{jk}{}^{S}\tilde{G}_{(n+1)(n+1)}\\
&+\tilde{\Gamma}^{n+1}_{(n+1)(n+1)}{}^{S}\tilde{G}_{(n+1)k}+\tilde{\Gamma}^{\gamma}_{(n+1)k}{}^{S}\tilde{G}_{(n+1)\gamma} + \epsilon\overline{\psi}^2\overline{g}^{ij}\big( \tilde{\Gamma}^{n+1}_{ij}{}^{S}\tilde{G}_{(n+1)k}+\tilde{\Gamma}^{n+1}_{jk}{}^{S}\tilde{G}_{i(n+1)}\\
&-\frac{\epsilon}{\overline{\psi}^2}\tilde{\Gamma}^l_{ij}\overline{g}_{lk}{}^{S}\tilde{G}_{(n+1)(n+1)}-\frac{\epsilon}{\overline{\psi}^2}\tilde{\Gamma}^l_{jk}\overline{g}_{il}{}^{S}\tilde{G}_{(n+1)(n+1)} \big)
\end{split}
\end{align}
Also setting $\sigma=n+1$ in (\ref{lemma1eq1}) we get:
\begin{align}\label{lemma1eq3}
\begin{split}
\frac{\partial}{\partial y}{}^{S}\tilde{G}_{(n+1)(n+1)}&=-\epsilon\overline{\psi}^2\overline{g}^{ij}\partial_j{}^{S}\tilde{G}_{i(n+1)}+2\tilde{\Gamma}^{\gamma}_{(n+1)(n+1)}{}^{S}\tilde{G}_{\gamma(n+1)}+\epsilon\overline{\psi}^2\overline{g}^{ij}\big( \tilde{\Gamma}^{\gamma}_{ij}{}^{S}\tilde{G}_{\gamma(n+1)}+\tilde{\Gamma}^{n+1}_{j(n+1)}{}^{S}\tilde{G}_{i(n+1)}\\
&-\frac{\epsilon}{\overline{\psi}^2}\tilde{\Gamma}^{l}_{j(n+1)}\overline{g}_{il}{}^{S}\tilde{G}_{(n+1)(n+1)} \big)
\end{split}
\end{align}
So we get that (\ref{lemma1eq1}) is equivalent to the system of PDE  formed by the equations (\ref{lemma1eq2}) and (\ref{lemma1eq3}), which are linear homogeneous equations on $^{S}\tilde{G}_{(n+1)\sigma}$ which can be written in the following form:
\begin{align}\label{lemma1eq4}
\frac{\partial}{\partial y}{}^{S}\tilde{G}_{(n+1)\sigma}=\mathcal{U}_{\sigma}(x,y,^{S}\tilde{G}_{(n+1)\beta},\partial_j\tilde{G}_{(n+1)\beta})
\end{align}
and under our hypothesis the functions on the right hand side are analytic functions on some neighborhood of the origin in $\mathbb{R}^{n+1}$. Also under our hypothesis we have that, not only this set of equations are satisfied, but they also satisfy the following initial data:
\begin{align}
^{S}\tilde{G}_{(n+1)\sigma}(x,0)=0
\end{align}
Now we know that the Cauchy-Kovalevskaya theorem asserts that this system admits just one set of analytic solutions satisfying these initial data, and since the system is homogeneous, we know that the trivial solution $^{S}\tilde{G}_{(n+1)\sigma}=0$ is such a solution, then this is the only solution. Hence the functions $^{S}\tilde{G}_{(n+1)\sigma}$ are actually zero on a neighborhood of the origin in $\mathbb{R}^{n+1}$ and this finishes the proof.
\end{proof}

First, we shall show that (\ref{edp4}) and (\ref{edp1}) have a solution in a
neighborhood of $0\in\mathbb{R}^{n+1}$. In order to do this we need to write
down these equations explicitly. From (\ref{riccicomponents}) we find that:
\begin{align*}
{}^{S}\tilde{R}_{ij}=  &  -\frac{\epsilon}{\overline{\psi}}\partial
_{n+1}l_{ij}+{}^{S}\overline{R}_{ij}+\epsilon\overline{g}^{kl}(l_{ij}%
l_{kl}-2l_{ki}l_{jl})+\frac{1}{\overline{\psi}}\overline{\nabla}_{j}%
\overline{\nabla}_{i}\overline{\psi}-\frac{1}{4}(\overline{\nabla}%
_{j}\overline{\omega}_{i}+\overline{\nabla}_{i}\overline{\omega}_{j})+\frac
{1}{4}\omega_{i}\overline{\omega}_{j}\\
&  -\frac{1}{2\overline{\psi}}(\overline{\omega}_{i}\partial_{j}\overline
{\psi}+\overline{\omega}_{j}\partial_{i}\overline{\psi}-\epsilon\tilde{\omega
}_{n+1}l_{ji}).
\end{align*}
By using the fact that
\[
l_{ij}=-\frac{1}{2\overline{\psi}}\partial_{n+1}\overline{g}_{ij}+\frac
{1}{2\overline{\psi}}\tilde{\omega}_{n+1}\overline{g}_{ij},
\]
we can write (\ref{edp1}) in the form
\begin{equation}%
\begin{split}
\frac{\epsilon}{2\overline{\psi}^{2}}\frac{\partial^{2}\overline{g}_{ij}%
}{\partial y^{2}}=  &  -{}^{S}\overline{R}_{ij}-\epsilon\overline{g}%
^{kl}(l_{ij}l_{kl}-2l_{ki}l_{jl})-\frac{1}{\overline{\psi}}\overline{\nabla
}_{j}\overline{\nabla}_{j}\overline{\psi}+\frac{1}{4}(\overline{\nabla}%
_{j}\overline{\omega}_{i}+\overline{\nabla}_{i}\overline{\omega}_{j})-\frac
{1}{4}\overline{\omega}_{i}\overline{\omega}_{j}\\
&  +\frac{1}{2\overline{\psi}}(\overline{\omega}_{i}\partial_{j}\overline
{\psi}+\overline{\omega}_{j}\partial_{i}\overline{\psi}-\epsilon\tilde{\omega
}_{n+1}l_{ji})+\frac{\epsilon}{2\overline{\psi}^{2}}(\overline{g}_{ij}%
\frac{\partial}{\partial y}\tilde{\omega}_{n+1}+\tilde{\omega}_{n+1}%
\frac{\partial}{\partial y}\overline{g}_{ij})\\
&  +\frac{\epsilon}{2\overline{\psi}^{3}}\frac{\partial}{\partial y}%
\overline{\psi}(\frac{\partial}{\partial y}\overline{g}_{ij}-\tilde{\omega
}_{n+1}\overline{g}_{ij})
\end{split}
\label{CK1}%
\end{equation}
On the other hand, (\ref{edp4}) is equivalent to:
\begin{align}
\label{gauge1}\tilde{g}^{\mu\lambda}\tilde{\nabla}_{\lambda}\tilde{\nabla
}_{\nu}\tilde{\omega}_{\mu}-\tilde{g}^{\mu\lambda}\tilde{\nabla}_{\lambda
}\tilde{\nabla}_{\mu}\tilde{\omega}_{\nu}=0.
\end{align}
Thus, from the compatibility condition we can rewrite the first term as
\[
\tilde{g}^{\mu\lambda}\tilde{\nabla}_{\lambda}\tilde{\nabla}_{\nu}%
\tilde{\omega}_{\mu}=\tilde{\nabla}_{\lambda}\tilde{\nabla}_{\nu}\tilde
{\omega}^{\lambda}-\tilde{\omega}^{\mu}\tilde{\omega}_{\mu}\tilde{\omega}%
_{\nu}+\tilde{\omega}^{\lambda}\tilde{\nabla}_{\lambda}\tilde{\omega}_{\nu
}+\tilde{g}^{\mu\lambda}\tilde{\nabla}_{\lambda}\tilde{\omega}_{\mu}%
\tilde{\omega}_{\nu}+\tilde{\omega}^{\mu}\tilde{\nabla}_{\nu}\tilde{\omega
}_{\mu}.
\]
From the definition of the curvature tensor we have
\[
\tilde{\nabla}_{\lambda}\tilde{\nabla}_{\nu}\tilde{\omega}^{\lambda}=\tilde
{R}^{\lambda}{}_{\sigma\nu\lambda}\tilde{\omega}^{\sigma}+\tilde{\nabla}_{\nu
}\tilde{\nabla}_{\lambda}\tilde{\omega}^{\lambda},
\]
that is,
\[
\tilde{\nabla}_{\lambda}\tilde{\nabla}_{\nu}\tilde{\omega}^{\lambda}%
=\tilde{\nabla}_{\nu}\tilde{\nabla}_{\lambda}\tilde{\omega}^{\lambda}%
-\tilde{R}_{\sigma\nu}\tilde{\omega}^{\sigma}.
\]
In this way, we get
\[
\tilde{g}^{\mu\lambda}\tilde{\nabla}_{\lambda}\tilde{\nabla}_{\nu}%
\tilde{\omega}_{\mu}=\tilde{\nabla}_{\nu}\tilde{\nabla}_{\lambda}\tilde
{\omega}^{\lambda}-\tilde{R}_{\sigma\nu}\tilde{\omega}^{\sigma}+\tilde{\omega
}^{\lambda}\tilde{\nabla}_{\lambda}\tilde{\omega}_{\nu}+\tilde{g}^{\mu\lambda
}\tilde{\nabla}_{\lambda}\tilde{\omega}_{\mu}\tilde{\omega}_{\nu}%
+\tilde{\omega}^{\mu}\tilde{\nabla}_{\nu}\tilde{\omega}_{\mu}-\tilde{\omega
}^{\mu}\tilde{\omega}_{\mu}\tilde{\omega}_{\nu}.
\]
Using this in (\ref{gauge1}) we obtain
\begin{align}
\label{gauge2}\tilde{g}^{\mu\lambda}\tilde{\nabla}_{\lambda}\tilde{\nabla
}_{\mu}\tilde{\omega}_{\nu}-\tilde{\nabla}_{\nu}\tilde{\nabla}_{\lambda}%
\tilde{\omega}^{\lambda}-\tilde{\omega}^{\lambda}\tilde{\nabla}_{\lambda
}\tilde{\omega}_{\nu}-\tilde{\omega}^{\mu}\tilde{\nabla}_{\nu}\tilde{\omega
}_{\mu}-\tilde{g}^{\mu\lambda}\tilde{\nabla}_{\lambda}\tilde{\omega}_{\mu
}\tilde{\omega}_{\nu}+\tilde{\omega}^{\mu}\tilde{\omega}_{\mu}\tilde{\omega
}_{\nu}+\tilde{R}_{\sigma\nu}\tilde{\omega}^{\sigma}=0.
\end{align}
These equations are equivalent to (\ref{gauge1}). Unfortunately, they cannot
be written in a form where we can apply the Cauchy-Kovalevskaya theorem.
However, if we consider these equations in the \textit{Lorentz gauge}
$\tilde{\nabla}_{\lambda}\tilde{\omega}^{\lambda}=0,$ we can show that the
resulting set of \textit{reduced equations} can be cast in the form required
by this theorem. Now, writing these equations explicitly we get
\begin{align}
\label{CK2}%
\begin{split}
\frac{\epsilon}{\overline{\psi}^{2}}\frac{\partial^{2}\tilde{\omega}_{\nu}%
}{\partial y^{2}}=  &  -\overline{g}^{ij}\tilde{\nabla}_{i}\tilde{\nabla}%
_{j}\tilde{\omega}_{\nu}+\tilde{\omega}^{\lambda}\tilde{\nabla}_{\lambda
}\tilde{\omega}_{\nu}+\tilde{\omega}^{\mu}\tilde{\nabla}_{\nu}\tilde{\omega
}_{\mu}+\tilde{g}^{\mu\lambda}\tilde{\nabla}_{\lambda}\tilde{\omega}_{\mu
}\tilde{\omega}_{\nu}-\tilde{\omega}^{\mu}\tilde{\omega}_{\mu}\tilde{\omega
}_{\nu}-\tilde{R}_{\sigma\nu}\tilde{\omega}^{\sigma}\\
&  +\frac{\epsilon}{\overline{\psi}^{2}}\big(\frac{\partial\tilde{\Gamma
}_{(n+1)\nu}^{\sigma}}{\partial y}\tilde{\omega}_{\sigma}+\tilde{\Gamma
}_{(n+1)\nu}^{\sigma}\frac{\partial\tilde{\omega}_{\sigma}}{\partial y}%
+\tilde{\Gamma}_{(n+1)\nu}^{\beta}\frac{\partial\tilde{\omega}_{\beta}%
}{\partial y}-\tilde{\Gamma}_{(n+1)\nu}^{\beta}\tilde{\Gamma}_{(n+1)\beta
}^{\sigma}\tilde{\omega}_{\sigma}\\
&  +\tilde{\Gamma}_{(n+1)(n+1)}^{\beta}\partial_{\beta}\tilde{\omega}_{\nu
}-\tilde{\Gamma}_{(n+1)(n+1)}^{\beta}\tilde{\Gamma}_{\beta\nu}^{\sigma}%
\tilde{\omega}_{\sigma}\big)
\end{split}
\end{align}

We shall regard these equations together with (\ref{CK1}) as a system of PDEs
for $(\tilde{g},\tilde{\omega})$. It is important to remark that (\ref{CK2})
depends on $\frac{\partial^{2}\overline{g}_{ij}}{\partial y^{2}}$ through
terms such as $\frac{\partial\tilde{\Gamma}_{(n+1)\nu}^{\sigma}}{\partial
y}\tilde{\omega}_{\sigma}$ or $\tilde{R}_{\sigma\nu}\tilde{\omega}^{\sigma}$.
But, as we are regarding (\ref{CK1}) and (\ref{CK2}) as a system, we just
replace $\frac{\partial^{2}\overline{g}_{ij}}{\partial y^{2}}$ in (\ref{CK2})
using (\ref{CK1}). Thus, if we consider that $\overline{\psi}$ is a given
analytic function in a neighborhood of the origin of $\mathbb{R}^{n+1}$ which
satisfies $\overline{\psi}\neq0$ in this neighborhood, then (\ref{CK1}) and
(\ref{CK2}) yield a system in the form%

\begin{equation}%
\begin{split}
\frac{\partial^{2}\overline{g}_{ij}}{\partial y^{2}}  &  =F_{ij}%
(x,y,\overline{g}_{ij},\tilde{\omega}_{\alpha},\partial_{\alpha}\overline
{g}_{ij},\partial_{\beta}\tilde{\omega}_{\alpha},\partial_{i\alpha}%
\overline{g}_{ij},\partial_{i\beta}\tilde{\omega}_{\alpha})\;\;\;1\leq i<j\leq
n\;\;;\;\;\alpha,\beta=1,\dots,n+1\\
\frac{\partial^{2}\tilde{\omega}_{\beta}}{\partial y^{2}}  &  =\mathcal{U}%
_{\beta}(x,y,\overline{g}_{ij},\tilde{\omega}_{\alpha},\partial_{\alpha
}\overline{g}_{ij},\partial_{\beta}\tilde{\omega}_{\alpha},\partial_{i\alpha
}\overline{g}_{ij},\partial_{ij}\tilde{\omega}_{\alpha})\;\;\;1\leq i<j\leq
n\;\;;\;\;\alpha,\beta=1,\dots,n+1\\
&
\end{split}
\label{sis}%
\end{equation}
Therefore, if we choose a specific order for the $\frac{n(n+1)}{2}$ components
of $\overline{g}_{ij}$ and the $n+1$ components of $\tilde{\omega}_{\beta}$,
then (\ref{sis}) may be regarded as a system of $\frac{(n+1)(n+2)}{2}$ PDEs
for the $\frac{(n+1)(n+2)}{2}$ functions $(\overline{g}_{ij},\tilde{\omega
}_{\beta})$. For such a system we give the following initial data
\begin{equation}%
\begin{split}
\overline{g}_{ik}(x,0)  &  =g_{ik}(x)\;\;1\leq i<k\leq n\\
\tilde{\omega}_{\beta}(x,0)  &  =\omega_{\beta}(x)\;\;\beta=1,\dots,n+1\\
&
\end{split}
\label{initialdata1}%
\end{equation}%
\begin{equation}%
\begin{split}
\frac{\partial\overline{g}_{ik}}{\partial y}(x,0)  &  =-2\overline{\psi
}(x,0)\Omega_{ik}(x)++g_{ik}(x)\omega_{n+1}(x)\doteq g_{ik}^{^{\prime}%
}(x)\;\;1\leq i<k\leq n\\
\frac{\partial\tilde{\omega}_{\beta}}{\partial y}(x,0)  &  =\omega_{\beta
}^{^{\prime}}(x)\;\;\beta=1,\dots,n+1
\end{split}
\label{initialdata2}%
\end{equation}
where $\omega_{\beta},\omega_{\beta}^{^{\prime}},\Omega_{ik}$ and $g_{ik}$ are
all analytic functions at $0\in\mathbb{R}^{n}$, and it is required that the
initial data $g_{ik}$ also satisfy that the condition $det(g_{ik})(0)\neq0$.
It is important to note that the right-hand side of (\ref{sis}) consists of
rational functions of the variables $\overline{g}_{ij},\tilde{\omega}_{\alpha
},\partial_{\alpha}\overline{g}_{ij},\partial_{\beta}\tilde{\omega}_{\alpha
},\partial_{a\alpha}\overline{g}_{ij},\partial_{a\beta}\tilde{\omega}_{\alpha
}$, and all the denominators are just $det(\overline{g}_{ik})$. On the other
hand, it follows from the initial data that $det(\overline{g}_{ik}%
)(0,0)=det(g_{ik})(0)\neq0$. Thus, since $det(\overline{g}_{ik})$ is a
polynomial of the functions $\overline{g}_{ik}$, and we know that for
$\overline{g}_{ik}^{\circ}\doteq\overline{g}_{ik}(0,0)$ this polynomial is
different from zero, then there is a neighborhood of $(\overline{g}%
_{ik}^{\circ})$ where this polynomial does not vanish. Using both this fact
and that the functions $F_{ij}$ and $\mathcal{U}_{\beta}$ are just these
rational functions multiplied by some power of $\overline{\psi}$, which, in
turn, is an analytic function in a neighborhood of $0\in\mathbb{R}^{n+1}$ and
$\overline{\psi}\neq0$ in this neighborhood, we then see that $F_{ij}$ and
$\mathcal{U}_{\beta}$ are analytic functions at $P=(0,0,\overline{g}%
_{ij}(0,0),\tilde{\omega}_{\alpha}(0,0),\partial_{\alpha}\overline{g}%
_{ij}(0,0),\partial_{\beta}\tilde{\omega}_{\alpha}(0,0),\partial_{a\alpha
}\overline{g}_{ij}(0,0),\partial_{a\beta}\tilde{\omega}_{\alpha}(0,0))$. We
can thus use the Cauchy-Kovalevskaya theorem to guarantee the existence of
solutions of the system (\ref{sis}) with the initial data (\ref{initialdata1}%
)-(\ref{initialdata2}). It is important to remark that, since the
Cauchy-Kovalevskaya theorem guarantees the existence of analytic solutions in
a neighbourhood of $0\in\mathbb{R}^{n+1}$, and as $det(\overline{g}%
_{ij})(0,0)\neq0$ from the initial data, then by continuity we know that there
exists a neighborhood of $0\in\mathbb{R}^{n+1}$ where $det(\overline{g}%
_{ij})(x,y)\neq0$.

In this way we have constructed a set of analytic solutions of the reduced
equations (\ref{sis}). In order for these solutions to satisfy the original
set of equations (\ref{edp4})-(\ref{edp1}) we need to show that they satisfy
the \textit{gauge condition} $\tilde{\nabla}_{\lambda}\tilde{\omega}^{\lambda
}=0$. With this in mind we present the following lemma.

\begin{lemma}
Consider $n+1=4$ and suppose that $(\tilde{g},\tilde{\omega})$ is an analytic
solution of (\ref{sis}) satisfying the initial data (\ref{initialdata1}%
)-(\ref{initialdata2}), and also assume that
\begin{align}
\tilde{\nabla}_{\lambda}\tilde{\omega}^{\lambda}|_{\Sigma_{0}}%
=0\label{gaugeconstraint1}\\
\frac{\partial\tilde{\nabla}_{\lambda}\tilde{\omega}^{\lambda}}{\partial
y}|_{\Sigma_{0}}=0\label{gaugeconstraint2}%
\end{align}
Then, $(\tilde{g},\tilde{\omega})$ satisfy the complete system of equations
(\ref{edp4})-(\ref{edp1}).
\end{lemma}

\begin{proof}
If $dim(\tilde{M})=4$ then we have seen that the following identity is satisfied on $\tilde{M}$:
\begin{align*}
\tilde{g}^{\alpha\beta}\tilde{\nabla}_{\alpha}(\tilde{g}^{\mu\nu}\tilde{\nabla}_{\mu}F_{\beta\nu})+\tilde{\omega}^{\beta}\tilde{g}^{\mu\nu}\tilde{\nabla}_{\mu}F_{\beta\nu}=0.
\end{align*}
Also if $(\tilde{g},\tilde{\omega})$ satisfy the reduced equations, then from (\ref{gauge2}) we get that
\begin{align}\label{gaugeprop1}
\tilde{g}^{\mu\nu}\tilde{\nabla}_{\mu}F_{\beta\nu}=-\tilde{\nabla}_{\beta}\tilde{\nabla}_{\lambda}\tilde{\omega}^{\lambda}
\end{align}
Then the previous identity gives us the following:
\begin{align*}
\tilde{g}^{\alpha\beta}\tilde{\nabla}_{\alpha}\tilde{\nabla}_{\beta}(\tilde{\nabla}_{\lambda}\tilde{\omega}^{\lambda})+\tilde{\omega}^{\beta}\tilde{\nabla}_{\beta}(\tilde{\nabla}_{\lambda}\tilde{\omega}^{\lambda})=0.
\end{align*}
It is not difficult to show that this is a second order linear and homogeneous equation for the function $\tilde{\nabla}_{\lambda}\tilde{\omega}^{\lambda}$ which can be rewritten as follows:
\begin{align}
\frac{\partial^2(\tilde{\nabla}_{\lambda}\tilde{\omega}^{\lambda})}{\partial y^2}=\mathcal{F}(x,y,\tilde{\nabla}_{\lambda}\tilde{\omega}^{\lambda},\partial_{\alpha}(\tilde{\nabla}_{\lambda}\tilde{\omega}^{\lambda}),\partial_{i\alpha}(\tilde{\nabla}_{\lambda}\tilde{\omega}^{\lambda}))
\end{align}
where the right-hand side is an analytic function at the origin. Then the Cauchy-Kovalevskaya theorem guarantees the existence of a unique solution for this equation satisfying the initial data (\ref{gaugeconstraint1})-(\ref{gaugeconstraint2}). Since $\tilde{\nabla}_{\lambda}\tilde{\omega}^{\lambda}=0$ satisfies all these requirements, we get that this is the unique solution. Using this in (\ref{gaugeprop1}) we see that under these conditions $(\tilde{g},\tilde{\omega})$ satisfies the full system (\ref{edp4})-(\ref{edp1}).
\end{proof}

Using this lemma, which only works in the 4-dimensional case, we see that we
should look at (\ref{gaugeconstraint1}) and (\ref{gaugeconstraint2}) as
additional constraints. Thus, our system of constraint equations consists of
the equations (\ref{edp2})-(\ref{edp3}) on the hypersurface $\Sigma_{0},$
together with the equations (\ref{gaugeconstraint1})-(\ref{gaugeconstraint2}).
This system is posed for the second fundamental form of $\Sigma_{0}$,
$\Omega_{ij}$, and the initial data $\omega_{\beta}^{\prime},$ and will be
referred as the \textit{Weyl constraints equations}. We shall denote the
initial data set by $(\Sigma_{0},g,\omega,\Omega,\omega^{\prime})$, where
$(\Sigma_{0},g,\omega)$ gives the Weyl structure of the hypersurface
$\Sigma_{0}$. With these notations, we can state the following theorem:

\begin{thm}
Let $(\Sigma_{0},g,\omega,\Omega,\omega^{\prime})$ be an initial data set
satisfying the Weyl constraint equations. Then $(\Sigma_{0},g,\omega)$ admits
a local analytic isometric embedding around $p\in\Sigma_{0}$ in a Weyl
manifold $(\tilde{M}^{4},\tilde{g},\tilde{\omega})$ such that the symmetric
part of the Ricci tensor of the embedding manifold vanishes.
\end{thm}

Using this theorem, we see that in order to guarantee the existence of an
isometric embedding of $(M^{3},g,\omega)$ at $p\in M^{3}$ in a Weyl manifold
$(\tilde{M}^{4},\tilde{g},\tilde{\omega})$ having vanishing symmetric part of
its Ricci tensor, we just need to show that we can always find an initial data
set $(M^{3},g,\omega,\Omega,\omega^{\prime})$ satisfying the Weyl constraint
equations in a neighborhood of $p\in M^{3}$. When dealing with these
constraints, we shall make use of the \textit{gauge} freedom we have in Weyl's
geometry. We have already seen that if we can construct an embedding for some
element $(M,g,\omega)\in$ $(M,[g],[\omega])$ in $(\tilde{M},\tilde{g}%
,\tilde{\omega})\in$ $(\tilde{M},[\tilde{g}],[\tilde{\omega}])$ then we can
construct an embedding for each element of $(M,[g],[\omega])$ in some element
of $(\tilde{M},[\tilde{g}],[\tilde{\omega}])$. Thus, we shall select a
particular element of $(M,[g],[\omega])$ where (\ref{gaugeconstraint1}) is
satisfied. Let us show that we can always do this. First, consider that
$(\overline{g}_{ij},\tilde{\omega}_{\beta})$ is a solution of (\ref{sis}) in a
neighborhood $U$ of $0\in\mathbb{R}^{n+1}$ satisfying the initial data
(\ref{initialdata1})-(\ref{initialdata2}). Then, defining
\[
\tilde{g}\doteq\overline{g}_{ij}dx^{i}\otimes dx^{j}+\epsilon\overline{\psi
}^{2}dy\otimes dy,
\]
we have that $(U,\tilde{g},\tilde{\omega})$ is a well-defined Weyl manifold.
Under these assumptions note that
\[
\tilde{\nabla}_{\lambda}\tilde{\omega}^{\lambda}=\overline{\nabla}_{k}%
\tilde{\omega}^{k}+\tilde{\nabla}_{n+1}\tilde{\omega}^{n+1},
\]
hence
\[
\tilde{\nabla}_{\lambda}\tilde{\omega}^{\lambda}|_{\Sigma_{0}}=\nabla
_{k}\omega^{k}+\tilde{\nabla}_{n+1}\tilde{\omega}^{n+1}(x,0).
\]
This last expression only depends on the initial data $(g,\omega)$ and
$\tilde{\omega}_{n+1}|_{\Sigma_{0}},\partial_{n+1}\tilde{\omega}%
_{n+1}|_{\Sigma_{0}}$. We now make the Weyl transformation
\begin{align*}
g  &  \rightarrow e^{-f}g\\
\omega_{k}  &  \rightarrow\omega_{k}-\partial_{k}f
\end{align*}
for some analytic function $f$. Then, for $(e^{-f}g,\omega-df)$
(\ref{gaugeconstraint1}) is equivalent to the equation
\begin{align}
\label{gaugeconstraint1.1}g^{ku}\nabla_{k}\nabla_{u}f+g^{ku}(\omega_{k}%
-\nabla_{k}f)(\omega_{u}-\nabla_{k}f)+g^{ku}\nabla_{k}\omega_{u}%
+e^{f}(\partial_{y}\tilde{\omega}^{4}+\tilde{\Gamma}_{4\sigma}^{4}%
\tilde{\omega}^{\sigma})|_{y=0}=0,
\end{align}
where $\tilde{\Gamma}_{4\sigma}^{4}=\frac{1}{\overline{\psi}}\partial_{\sigma
}\overline{\psi}-\frac{1}{2}\tilde{\omega}_{\sigma}$. Since $\overline{\psi}$
is considered as a given analytic function and both $\tilde{\omega}%
_{4}|_{\Sigma_{0}}$ and $\partial_{y}\tilde{\omega}_{4}|_{\Sigma_{0}}$ are
also arbitrary given analytic functions, then (\ref{gaugeconstraint1.1}) is a
second-order PDE for the function $f$. Thus, from now on, we shall regard
$\overline{\psi}|_{\Sigma_{0}}=\psi(x)$, $\tilde{\omega}_{4}|_{\Sigma_{0}%
}\doteq\omega_{4}(x)$ and $\partial_{y}\tilde{\omega}_{4}|_{\Sigma_{0}}%
\doteq\eta(x)$ as given analytic functions, which will be involved in the
initial data of the system (\ref{sis}). Then, we can guarantee the existence
of an analytic solution for (\ref{gaugeconstraint1.1}). To see this, we can
use a coordinate system $(x^{i})$ on $M$ around $p$, satisfying that
$g_{1k^{\prime}}=0$, with $k^{\prime}=2,3$. In this way,
(\ref{gaugeconstraint1.1}) can be cast in the form
\begin{align*}
g^{11}\nabla_{1}\nabla_{1}f  &  +g^{k^{\prime}u^{\prime}}\nabla_{k^{\prime}%
}\nabla_{u^{\prime}}f+g^{ku}(\omega_{k}-\nabla_{k}f)(\omega_{u}-\nabla
_{k}f)+g^{ku}\nabla_{k}\omega_{u}\\
&  +e^{f}\{\frac{\epsilon}{\psi^{2}}\eta-\frac{2\epsilon}{\psi^{3}}%
\frac{\partial\overline{\psi}}{\partial y}|_{\Sigma_{0}}\omega_{4}+(\frac
{1}{\psi}\partial_{k}\psi-\frac{1}{2}\omega_{k})\omega^{k}+\frac{\epsilon
}{\psi}(\frac{1}{\psi}\frac{\partial}{\partial y}\overline{\psi}|_{\Sigma_{0}%
}-\frac{1}{2}\omega_{4})\omega_{4}\}=0,
\end{align*}
where $g^{11}\neq0$ in a neighborhood of the origin, $k^{\prime},u^{\prime
}=2,3$, and all the known quantities involved are analytic in a neighborhood
of $0\in\mathbb{R}^{n}$. Therefore, this last equation has the form
\[
\frac{\partial^{2}f}{\partial(x^{1})^{2}}=\mathcal{U}(x,\partial_{i}%
f,\partial_{k^{\prime}i}f),
\]
where the right-hand side is analytic at the origin, and hence the
Cauchy-Kovalevskaya theorem guarantees the existence of an analytic solution.
We thus have shown that, given a Weyl manifold $(M^{3},g,\omega)$ we can
always find an element of $(M^{3},[g],[\omega])$ for which
(\ref{gaugeconstraint1}) is satisfied. Hence there is no loss of generality in
assuming that $(M^{3},g,\omega)$ satisfies this condition. In this way, we can
reduce the Weyl constraint equations to the following set of equations:%

\begin{align}
\epsilon{\psi}g^{kl}(\nabla_{k}\Omega_{il}-\nabla_{i}\Omega_{kl})+\frac{\psi
}{2}g^{kl}(\omega_{i}\Omega_{kl}-\omega_{k}\Omega_{il})+\frac{n-1}{4}%
(\partial_{i}\eta-\omega_{i}^{^{\prime}})=0\label{reducedconstraints1}%
\end{align}
\begin{align}
g^{ij}g^{kl}(R_{kilj}+\epsilon(\Omega_{ij}\Omega_{kl}-\Omega_{ki}\Omega
_{jl}))=0\label{reducedconstraints2}%
\end{align}
\begin{align}
\frac{\partial\tilde{\nabla}_{\lambda}\tilde{\omega}^{\lambda}}{\partial
y}|_{\Sigma_{0}}=0.\label{reducedconstraints3}%
\end{align}
Equations (\ref{reducedconstraints1}) and (\ref{reducedconstraints2}) are
dealt with in the same way it was done in \cite{dahia1}, just carrying the
extra terms along the same computations. Doing this, from
(\ref{reducedconstraints2}) we find an explicit expression for $\Omega_{11}$
in terms of the other variables and from (\ref{reducedconstraints1}) a set
first-order PDEs of the Cauchy-Kovalevskaya-type for the functions
$\Omega_{1k^{\prime}}$ and $\Omega_{r^{\prime}3}$, where $k^{\prime}=2,3$ and
$r^{\prime}$ is fixed, having either the value $2$ or $3$. In this system, the
remaining components of $\Omega_{ij}$ are set as arbitrary analytic functions.
In the same way we did when dealing with (\ref{gaugeconstraint1}), in this
procedure, a coordinate system on $M^{3}$ around $p$ is chosen such that
$g_{1k^{\prime}}=0$, and the coordinate $x^{1}$ is chosen as the variable with
respect to which we pose the constraint equations in the Cauchy-Kovalesvskaya
form. Now, we shall deal with the remaining equation
(\ref{reducedconstraints3}). First, let us write it down explicitly:
\begin{equation}
\frac{\partial(\tilde{\nabla}_{\lambda}\tilde{\omega}^{\lambda})}{\partial
y}=\frac{\partial(\partial_{k}\tilde{\omega}^{k})}{\partial y}+\frac
{\partial^{2}\tilde{\omega}^{4}}{\partial y^{2}}+\frac{\partial\tilde{\Gamma
}_{kj}^{k}}{\partial y}\tilde{\omega}^{j}+\tilde{\Gamma}_{kj}^{k}%
\frac{\partial\tilde{\omega}^{j}}{\partial y}+\frac{\partial\tilde{\Gamma
}_{4\sigma}^{4}}{\partial y}\tilde{\omega}^{\sigma}+\tilde{\Gamma}_{4\sigma
}^{4}\frac{\partial\tilde{\omega}^{\sigma}}{\partial y}.
\label{reducedconstraints4}%
\end{equation}
Using the following expressions for the connection components involved in the
previous expression
\begin{align*}
\tilde{\Gamma}_{4\beta}^{4}  &  =\frac{1}{\overline{\psi}}\partial_{\beta
}\overline{\psi}-\frac{1}{2}\tilde{\omega}_{\beta},\\
\tilde{\Gamma}_{ij}^{l}  &  =\overline{\Gamma}_{ij}^{l},
\end{align*}
together with the definition of the second fundamental form $\Omega,$
\[
\Omega_{ji}=\{-\frac{1}{2\overline{\psi}}\frac{\partial\overline{g}_{ij}%
}{\partial y}+\frac{1}{2\overline{\psi}}\overline{g}_{ij}\tilde{\omega}%
_{4}\}|_{\Sigma_{0}},
\]
we see that when we restrict (\ref{reducedconstraints4}) to $\Sigma_{0}$ the
last four terms depend on given data $g_{ij},\omega_{k},\eta,\omega
_{4}^{\prime},\psi$ and terms up to first-order in $\Omega_{ij}$ and the
remaining $\omega_{k}^{\prime}$. Also, since $(\tilde{g},\tilde{\omega})$
satisfy the reduced equations (\ref{sis}), then:
\[
\frac{\partial^{2}\tilde{\omega}_{4}}{\partial y^{2}}=\mathcal{U}%
_{n+1}(x,y,\overline{g}_{ij},\tilde{\omega}_{\alpha},\partial_{\alpha
}\overline{g}_{ij},\partial_{\beta}\tilde{\omega}_{\alpha},\partial_{i\alpha
}\overline{g}_{ij},\partial_{ai}\tilde{\omega}_{\alpha}),\;\;\;1\leq i<j\leq3
\; ; \; a=1,2,3.
\]
It follows that
\[
\frac{\partial^{2}\tilde{\omega}_{4}}{\partial y^{2}}|_{\Sigma_{0}%
}=\mathcal{U}_{n+1}^{\prime}(x,\Omega_{ij},\omega_{\alpha}^{\prime}%
,\partial_{a}\Omega_{ij}),\;\;\;1\leq i<j\leq3 \; ; \; a=1,2,3.
\]
Thus, constraining (\ref{reducedconstraints4}) to $\Sigma_{0}$ and setting the
left-hand side equal to zero, we get
\[
g^{ku}\partial_{k}\omega_{u}^{\prime}+\partial_{k}g^{ku}\omega_{u}^{\prime
}+\partial_{k4}\overline{g}^{ku}|_{\Sigma_{0}}\omega_{u}+\frac{\partial
\overline{g}^{ku}}{\partial y}|_{\Sigma_{0}}\partial_{k}\omega_{u}%
+\mathcal{O}(x,\Omega_{ij},\omega_{k}^{\prime},\partial_{k}\Omega_{ij})=0.
\]
Using the same special form of the metric in the coordinate system used to
study the constraints, we can rewrite this last equation as
\[
\frac{\partial\omega_{1}^{\prime}}{\partial x^{1}}=\mathcal{O}^{\prime
}(x,\Omega_{ij},\omega_{k}^{\prime},\partial_{k^{\prime}}\omega^{\prime
}_{j^{\prime}},\partial_{k}\Omega_{ij}) \; \; i,j,k=1,2,3 \;\; ; \;\;
j^{\prime},k^{\prime}=2,3.
\]
Then, we see that the constraint equations (\ref{reducedconstraints1}%
)-(\ref{reducedconstraints3}) can be written as a set of first-order PDEs of
the following form:
\begin{align}
\label{reducedconstraintsis}%
\begin{split}
\frac{\partial\Omega_{1k^{\prime}}}{\partial x^{1}}  &  =\mathcal{H}%
_{k^{\prime}}(x,\Omega_{1j^{\prime}},\omega_{1}^{\prime},\partial_{u^{\prime}%
}\Omega_{1j^{\prime}}),\;\;u^{\prime},j^{\prime},k^{\prime}=2,3\\
\frac{\partial\Omega_{3r^{\prime}}}{\partial x^{1}}  &  =\mathcal{H}%
_{r^{\prime}}(x,\Omega_{1j^{\prime}},\omega_{1}^{\prime},\partial_{u^{\prime}%
}\Omega_{1j^{\prime}}),\;\;u^{\prime},j^{\prime}=2,3\;\;;\;\;r^{\prime
}\;\;fixed\;\;with\;\;r^{\prime}=2\;\;or\;\;r^{\prime}=3\\
\frac{\partial\omega_{1}^{\prime}}{\partial x^{1}}  &  =\mathcal{O}^{\prime
}(x,\Omega_{1j^{\prime}},\omega_{1}^{\prime},\partial_{u^{\prime}}%
\Omega_{1j^{\prime}}),\;\;u^{\prime},j^{\prime}=2,3
\end{split}
\end{align}
together with an explicit algebraic expression for $\Omega_{11}$. In this set
up the rest of the $\Omega_{ij}$ and $\omega_{2}^{\prime},\dots,\omega
_{n}^{\prime}$ are set as given arbitrary analytic functions. The equations
(\ref{reducedconstraintsis}) are of the Cauchy-Kovalevskaya type and hence we
know that this system admits a solution. We now can state the main result of
this section.

\begin{thm}
Any 3-dimensional Weyl structure $(M^{3},[g],[\omega])$ admits a local
analytic isometric embedding at any point $p\in M^{3}$ in a 4-dimensional Weyl
structure $(\tilde{M}^{4},[\tilde{g}],[\tilde{\omega}])$ having vanishing
symmetric part of its Ricci tensor.
\end{thm}

\section{Final Comments}

In the present article, we have considered the embbeding problem in the
context of Weyl geometry and have proven that some of the
Campbell-Magaard-type theorems can be naturally extended from Riemannian to
Weyl's geometry, although some instances appear that are not exactly analogous
to their Riemannian counterpart. The investigation of embeddings in Weyl
manifolds has led us to discover an interesting and rather unexpected no-go
result in this direction, and to establish an important geometrical identity
which seems to be essential for studying embeddings in Weyl spaces, in
arbitrary dimensions, in which the symmetric part of its Ricci tensor
vanishes. We have worked out the embedding problem in the 3-dimensional case
and showed that this solution does not hold in other dimensions. We believe
that the complete solution of the general problem, still left open, may be
regarded as a mathematical motivation for studying other embedding problems in
the framework of Weyl's geometry, which may originate from modern theoretical
physics. Finally, we would like to mention that an extension of \cite{dahia3}
to the context of Weyl geometry is being studied at the moment. The results
coming from these further studies should be considered as a completion of the
present article.

\section*{Acknowledgements}

\noindent R. A and C. R. would like to thank CNPq and CLAF for financial support.

\end{document}